\documentclass[letterpaper, 12 pt, onecolumn]{ieeeconf}  

\IEEEoverridecommandlockouts                              
\usepackage{graphics,epsfig,epic,eepic,amssymb,latexsym,amsmath,amssymb,picinpar,epstopdf,times,mathptm}
\usepackage{color}

\usepackage{algorithm}
\usepackage{algpseudocode}
\usepackage{varwidth}
\usepackage{url}
\usepackage[colorinlistoftodos]{todonotes}
\newcommand{\rr}{\mathop{{\rm I}\mskip-4.0mu{\rm R}}\nolimits}
\newcommand{\Z}{\mathop{{\rm Z}\mskip-7.0mu{\rm Z}}\nolimits}

\newcommand{\Uf}{{\cal U}}
\newcommand{\Ef}{{\cal E}}

\newcommand{\Af}{{\cal A}}

\newcommand{\Df}{{\cal D}}
\newcommand{\Xf}{{\cal X}}

\newtheorem{theorem}{Theorem}

\newtheorem{definition}{Definition}

\newtheorem{proposition}{Proposition}

\newtheorem{assumption}{Assumption}

\newtheorem{remark}{\textbf{Remark}}

\title{Networked Constrained Cyber-Physical Systems subject to  malicious attacks: a resilient set-theoretic  control approach}

\author{
	\thanks{*Walter Lucia and Giuseppe Franz\`e are with the
		DIMES,
		Universit\'a degli Studi della Calabria, Via Pietro Bucci, Cubo
		42-C, Rende (CS), 87036, ITALY
		{\tt \small \{walter.lucia, giuseppe.franze\}@unical.it}} 
	Walter Lucia*\qquad
	\thanks{**Bruno Sinopoli is with the  Electrical and Computer Engineering Department, Carnegie Mellon University, Pittsburgh, PA 15213, USA, {\tt \small brunos@ece.cmu.edu}} 
	Bruno Sinopoli**\qquad
	 Giuseppe Franz\`e*  
}

	\date{ }
	
\begin{document}

\maketitle
\thispagestyle{empty}
\pagestyle{empty}



\begin{abstract}
In this paper a novel set-theoretic control framework for Networked Constrained Cyber-Physical Systems
 is presented. 
By  resorting to set-theoretic  ideas and  the physical watermarking concept,  an anomaly detector module and a  control remediation strategy are formally
derived with the aim to contrast severe cyber attacks affecting the communication channels. The resulting scheme ensures Uniformly Ultimate Boundedness  and constraints fulfillment regardless of any admissible 
attack scenario. Simulation results show the effectiveness of the proposed strategy both against  Denial of Service 
and False Data Injection 
attacks.
\end{abstract}

\section{INTRODUCTION}
Cyber-Physical Systems (CPSs) represent the integration of computation, networking, and physical processes that are expected to play a major role in the design and development of future engineering systems equipped with improved  capabilities ranging from autonomy to reliability and cyber security, see \cite{SaAn11} and references therein.   
%
The use of communication infrastructures  and heterogeneous IT components  have certainly improved scalability and functionality features in several applications
(transportation systems, medical technologies, water distributions, smart grids and so on), but on the other hand they have made  such systems highly 
vulnerable to cyber threats,
see e.g. the  attack on the network power transmission \cite{Gordman09} or the  Stuxnet warm  which infects  the  Supervision Control and Data Acquisition
system used to regulate  uranium enrichments  \cite{Chen10}.  
%
%
%
Recently, the analysis of the CPS security  from a theoretic perspective has received increasing attention and different solutions to discover cyber attack occurrences have been proposed, see \cite{MiPaPa13}, \cite{MoWeSi15}, \cite{PaDoBu13}, \cite{WeSi15}  and reference therein for detailed discussions.  
First,  it is important to underline that if  the attacker  and  defender share the same information then a passive anomaly detection system has no chance to identify stealthy attacks   \cite{WeSi15}.  There, the authors propose the introduction of an artificial  time-varying  model correlated to the CPS dynamics  so that any adversary attempting to manipulate the system state is revealed through its effect on such an extraneous time-varying system.  

Along these lines, a relevant approach is provided in \cite{MoWeSi15} where  the physical watermarking concept is exploited. Specifically,  a noisy control signal is 
superimposed  to a given optimal control input in order to authenticate the physical dynamics of the system. In \cite{TeShaSaJo12}, the authors  modify the system structure in order to reveal zero dynamic attacks, while    in \cite{MiPaPa13} a  coding sensor outputs is considered to detect FDI attacks. 
It is worth to point out that most of  works addressing CPSs  focus their attention only on the detection problem leaving out the control countermeasures. To the best of the author's knowledge very few control remediation strategies against cyber attacks have been proposed, see e.g.  \cite{FaTaDi14} where  a first contribution
for dealing with  
 CPS affected by corrupted sensors and actuators has been presented.

In this paper  two classes of cyber attacks will be analyzed:  i)\textit{ partial model knowledge  attacks}  and ii) \textit{full model knowledge  attacks} \cite{TeShaSaJo15}. The former is capable to break  encryption algorithms which  protect the communication channels and to modify the signals sent to the actuators and to the controller with the aim to cause physical damages.  
The second class can inject malicious data within  the control architecture. Zero-dynamics and FDI attacks fall into such a category \cite{TeShaSaJo15}.


In the sequel, the main aim is to develop a control architecture capable to manage constrained CPSs subject to malicious  data attacks. As one of its main merits, the strategy  is able to combine  into a unique framework detection/mitigation tasks with control purposes. In fact both the detection and control phases are addressed by using  the watermarking approach and the   set-theoretic  paradigm
 firstly introduced in \cite{BeRh71} and then successfully   applied in e.g. \cite{Bl08}, \cite{Angeli2008}, \cite{FraTeFa15}. \\
Specifically, the  identification module can be viewed as an active detector that, differently from the  existing solutions, does not require  neither input  or model manipulations. Moreover,  
a watermarking like behavior can be simply obtained
 during the on-line computation of the control action. 
The attack mitigation is  achieved by exploiting the concept of one-step controllable set jointly with 
cyber actions (communication disconnection, channels re-encryption) 
in order to ensure guaranteed control actions under any admissible attack scenario.\\
%
Finally, a  simulations campaign is provided under several attack scenarios to prove the effectiveness of the proposed methodology.

\section{PRELIMINARIES AND NOTATIONS}

Let us consider the class of  
Networked Constrained Cyber-Physical System (\textbf{NC-CPS}) described by the following discrete-time LTI model 
where we assume w.l.o.g. that the state vector is fully available:
%
\begin{equation}\label{eq:sys}
\begin{array}{rcl}
x(t+1)&=&A x(t) + B u (t) + B_d d_x(t)\\
y(t)&=&x(t)+d_y(t)
\end{array}
\end{equation}
where $t\in\Z_+:=\{0,1,...\},$ $x(t)\in \rr^n$ denotes the plant  state,   $u(t)\in\rr^{m}$ the control input, $y(t)\in \rr^n$ the output state measure and $d_x(t) \in \Df_x \subset \rr^{d_x},\, d_y(t) \in \Df_y \subset \rr^{n},\, \forall t \in \Z_+,$ exogenous bounded plant and measure disturbances, respectively. 
 Moreover (\ref{eq:sys}) is subject to the following state and input set-membership constraints:
\begin{equation}\label{eq:constraints}
u(t)\in \Uf,\quad  x(t)\in \Xf, \,\, \forall t \geq 0,
\end{equation}
\begin{definition}\label{UBB}
Let $S$ be a neighborhood of the origin. The closed-loop  trajectory of (\ref{eq:sys})-(\ref{eq:constraints})  is said to be Uniformly Ultimate Boundedness  (UUB) in $S$ if for all $\mu >0$ there exists $T(\mu)>0$ and  $u:=f(y(t))\in \Uf$ such that, for every $\|x(0)\| \leq \mu,$ $x(t)\in S$ $\forall d_x(t) \in \Df_x,\, \forall d_y(t) \in \Df_y,\, \forall t \geq T(\mu).$ 
\end{definition}
\begin{definition}\label{definition:one-step-controllable-sets}
	A set $\mathcal{T}\subseteq \rr^n$ is Robustly Positively Invariant (RPI) for (\ref{eq:sys})-(\ref{eq:constraints}) if there exists a control law $u:=f(y(t))\in \Uf$ such that,  once the closed-loop solution
	$x(t+1)=Ax(t)+Bf(y(t))+ B_dd_x$ 
	enters inside that set at any given time $t_0$, it remains in it for all future  instants, i.e. $x(t_0) \in \mathcal{T} \rightarrow x(t) 
	\in \Xi , \forall d_x(t) \in \Df_x,\,\forall d_y(t) \in \Df_y,\, \forall t \geq t_0.$ 
	\hfill $\Box$
\end{definition}
\begin{definition}
	Given the sets $\Af, \Ef\!\! \subset\!\! \rr^n$ $\Af\!\oplus\!\Ef\!\!:=\{a\!+\!e:\!\! \,a\!\in\Af, e\!\in\!\Ef\}$ is the {\em Minkowski Set Sum} and
%
	 $\Af\!\sim \!\Ef\!\!:=\{a\in \Af: \, a+e \in\Af,\,\forall e\in \Ef\}$  the {\em Minkowski Set Difference}.
	\hfill $\Box$
\end{definition}

\subsection{Set-theoretic receding horizon  control scheme (ST-RHC)}\label{section:dual-mode}

In the sequel, the receding  horizon control scheme proposed in \cite{Angeli2008} and based on the philosophy developed in the seminal paper  \cite{BeRh71}  is summarized.\\
Given the constrained LTI system (\ref{eq:sys})-(\ref{eq:constraints}),  determine a state-feedback $u(\cdot)=f(y(\cdot))\in \mathcal{U}$ capable  \textit{i)} to asymptotically stabilize (\ref{eq:sys}) and \textit{ii)} to drive  the state trajectory $x(\cdot)\in \mathcal{X}$ within a pre-specified  region $\mathcal{T}^0$ in a finite number of steps $N$ regardless of any disturbance realization $d_x(t)\in \Df_x,\,\,d_y(t)\in \Df_y.$

\noindent The latter can be addressed by resorting to the following receding horizon control  strategy:\\
\noindent {\bf{Off-line -}}
\begin{itemize}
\item Compute a   stabilizing state-feedback control law  $u^0(\cdot)=f^0(y(\cdot))$ complying with (\ref{eq:constraints}) 
and the associated RPI region $\mathcal{T}^0;$ 
\item 
Starting from $\mathcal{T}^0,$   determine a sequence of $N$ robust one-step ahead controllable sets $\mathcal{T}^i$ (see \cite{Bl08}): 
	\begin{equation} \label{eq:family-one-step}
	\begin{array}{l}
\mathcal{T}^0 :=  \mathcal{T}\\
\mathcal{T}^i  :=  \{ x \!\! \in \Xf:\forall d_x(t)\in \Df_x,\,d_y(t)\in \Df_y,\,\, \exists \, u \in \Uf:\\
		 \quad \qquad A (x+d_y(t)) + B u +B_dd_x(t) \in	{\mathcal{T}}^{i-1} \}\\
	 \,\,\,\,\,\,\,\,\,\,\,\,\,\,\,= \{ x \!\! \in \Xf: \exists \, u \in \Uf:  A x + B u   \in 
	\tilde{\mathcal{T}}^{i-1} \},
	i=1,\ldots, N
	\end{array}
	\end{equation}
where $\tilde{\mathcal{T}}^{i-1}:=\mathcal{T}^{i-1} \sim  B_d \Df_x \sim  A \Df_y.$ 
\end{itemize}

\noindent {\bf{On-line -}}

Let $x(0) \in \displaystyle\bigcup_{i=0}^N \mathcal{T}^{i},$ the command $u(t)$ is obtained as follows:
\begin{itemize}
	\item  Let $i(t) := \min \{i: y(t) \in \mathcal{T}^i \}$
	\item If $i(t)=0$ then $u(t)=f^0(y(t))$ else solve the following semi-definite programming (SDP) problem:
	\begin{equation}\label{fun_opt_b_3}
	{u}(t) =  \arg \min J_{j(t)}(y(t),u) \quad s.t.
	\end{equation}
	\begin{equation}\label{cond_opt_b_5}
	A x(t) + B u \in  {\tilde{\mathcal{T}}}^{i(t)-1},~ u \in \Uf 
	\end{equation}
	where $J_{j(t)}(y(t),u)$ is a cost function and $j(t)$ a time-dependent selection index.
\end{itemize}
\begin{remark}\label{remark:varaible_cost_function}
It is worth noticing that the cost function $J_{j(t)}(y(t),u)$ can be arbitrarily chosen without compromising the final objective of the control strategy and, in principle, it may be changed at each time instant.
\hfill $\Box$ 
\end{remark}
\section{PROBLEM FORMULATION}\label{section:problem_formulation}
\begin{figure}[!h]
\centering
\includegraphics[width=0.8\linewidth]{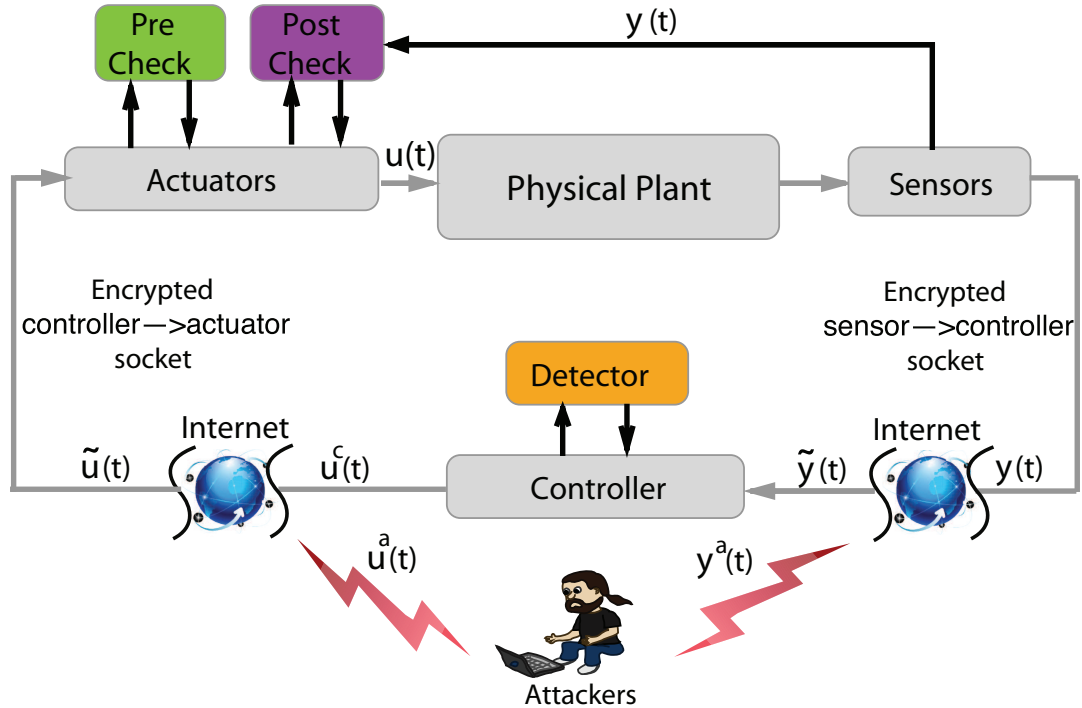}
\caption{Encrypted NC-CPS over Internet}
\label{fig:distribuited_control_architecture}
\end{figure}

In the sequel,  we consider CPSs  whose  physical plant is modeled as (\ref{eq:sys})-(\ref{eq:constraints}),
while the controller is  spatially distributed and a cyber median is used to build virtual communication channels from the plant to the controller and {\it{vice-versa}}, see  Fig. \ref{fig:distribuited_control_architecture}. 
We assume that   \textit{sensors-to-controller} and \textit{controller-to-actuators} communications are executed via Internet by means of encrypted sockets while all the remaining channels are local and externally not accessible. Moreover, malicious agents have the possibility to attack the communication over Internet by breaking the protocol security and may  compromise/alter  data flows in both the communication channels. Within such a context, two classes  of attacks will be taken into account: \textit{a)} Denials of Service (DoS)  and  \textit{b)} False Data Injections (FDI):   DoS attacks prescribe that attackers prevent  the standard sensor and controller data flows, 
while FDI occurrences give rise to arbitrary data injection on the relevant system signals, i.e. command inputs and state measurements.\\
%
\noindent Specifically, we shall  model attacks on the actuators as
\begin{equation}\label{acutator_attack_model}
\tilde{u}(t):=u^c(t)+u^a(t)
\end{equation}
where $u^c(t)\in \rr^m$ is the command input determined by the controller, $u^a(t)\in \rr^m$  the attacker perturbation and $\tilde{u}(t)\in\rr^m$  the resulting corrupted signal. Similarly, sensor  attacks has the following structure:
\begin{equation}\label{sensor_attack_model}
\tilde{y}(t):=y(t)+y^a(t)
\end{equation}
where 
$y^a(t)\in \rr^n$ is the attacker signal and $\tilde{y}(t)\in\rr^n$  the resulting corrupted measurement.  

%
%
\noindent From now on, the following assumptions   are made:
\begin{assumption}\label{ass1}
	{\it{An encrypted socked between controller and plant can be on-demand reestablished in at most $T_{encry}$  time instants.}}
\end{assumption}
\begin{assumption}\label{ass2}
	{\it{The minimum amount of time $T_{viol}$ required to violate the  cryptography algorithm is not vanishing, i.e   
	$T_{viol}\geq T_{encry}.$
			}}
\end{assumption}
\begin{assumption}\label{ass3}
 {\it{No relevant channel delays are due to the communication medium, i.e. all the  induced delays are  less than the  sampling time $T_c.$}}
\end{assumption}
\begin{remark}\label{remark:assumptions}
{\it{Assumption}} \ref{ass2}  relies on the fact that communication channels are not compromised for at least $T_{encry}$ time instants downline of
a new  encrypted socked is established. As a consequence,   the plant-controller structure  is guaranteed  w.r.t   the  sensor/actuator data truthfulness. 
\hfill $\Box$
\end{remark}
Then,  the problem we want to solve can be stated as follows:\\
\noindent \textbf{Resilient Control Problem of NC-CPSs subject to cyber attacks (\textbf{RC-NC-CPS})} - 
{\it
Consider  the control architecture of Fig. \ref{fig:distribuited_control_architecture}. Given  the \textbf{NC-CPS} model  (\ref{eq:sys})-(\ref{eq:constraints})  subject to DoS and/or FDI (\ref{acutator_attack_model})-(\ref{sensor_attack_model}) attacks , determine 
%
\begin{itemize}
	\item [ ] \textbf{-(P1)} An anomaly detector module $\mathbf{D}$ capable to discover  cyber attack occurrences;
	
	\item [ ] \textbf{-(P2)} A control strategy 
	$
	u(\cdot)=f(\tilde{y}(\cdot),\mathbf{D})
	$ 
such that  Uniformly Ultimate Boundedness  is ensured  and  prescribed constraints  fulfilled 	regardless of the presence of any admissible attack scenario. Moreover, if 
	$u^a(t)\equiv 0,\,y^a(t) \equiv 0$  (attack free scenario)  and $d_x(t), \,d_y(t)\equiv 0$ (disturbance free scenario)   $\forall t\geq \bar{t}$ then the regulated plant is asymptotically stable.
	\end{itemize}
}

\noindent The \textbf{RC-NC-CPS} problem will be addressed by properly customizing the  dual model set-theoretic control scheme described in Section \ref{section:dual-mode}. 

\section{SET-THEORETIC CHARACTERIZATION AND IDENTIFICATION OF ATTACKS}
In this section, an identification attack module  will be developed.  To this end, the following preliminaries are necessary.


First, notice that  according to  (\ref{fun_opt_b_3})-(\ref{cond_opt_b_5}) the following set-membership conditions hold true:
\begin{eqnarray}\label{robust-set-contraction}
x(t+1)&\in& \mathbf{Y}^+(y(t),u^c(t)) \nonumber \\
\mathcal{T}^{i(t)-1} &\supseteq& \mathbf{Y}^+(y(t),u^c(t))\!:=\!
\{ 
z\in \rr^n\!\!\!:\exists d_x\in \Df_x,\,d_y\!\!\in\!\! \Df_y,\nonumber \\
&&\qquad \qquad z\!=\!Ay(t)\!+\!B{u}^c\!\!+\!\!B_d d_x+\!\!d_y 
\!\}
\end{eqnarray}
with $\mathbf{Y}^+(y(t),u^c(t))$  the expected output prediction set.
Then, by using the classification given in \cite{TeShaSaJo15},
we consider   attackers having  the following disclosure and disrupt resources:
\begin{itemize}
	\item \textit{Disclosure}: An attacker can access to the  command inputs $u(t)$ and to the sensor measures $y(t);$
	\item \textit{Disrupt}: An attacker can inject arbitrary vectors  $u^a(t)$, $y^a(t)$ on the actuator and sensor communication channels but it cannot read and write on the same channel in a single time interval.
\end{itemize} 
Finally,  we consider   attacks belonging to the following categories:
%
\begin{definition}\label{def:stealthy-attack} 
Let us denote with $\mathcal{I}_a$ and $\mathbf{Y}_a^+$  the attacker model knowledge and expected output prediction set, respectively, then
an \textit{Attack with full model knowledge} is an attack with full information, $\mathcal{I}^{full},$ about the closed-loop dynamics of the physical plant, 
	\begin{equation}\label{eq:full_information}
	\mathcal{I}_a\equiv \mathcal{I}^{full}\!\!\!:=\!\!
	\left\{
    (\ref{eq:sys})-(\ref{eq:constraints}), f^0,\!\!
	\{\mathcal{T}^i\}_{i=0}^N,\,
	y(t),\,
	\mbox{opt:}(\ref{fun_opt_b_3})\!\!-\!\!(\ref{cond_opt_b_5}) 
	\right\}
	\end{equation}
and  perfect understanding of the expected output set,
$
	\mathbf{Y}_a^+\equiv \mathbf{Y}^+.
$
\end{definition}
\begin{definition}\label{def:no-stealthy-attack} 
	An \textit{Attack with partial model knowledge} is an attack with partial information, $\mathcal{I}_a$,  about the closed-loop dynamics of (\ref{eq:sys}), e.g.  
	\begin{equation}\label{eq:partial_information}
	\mathcal{I}_a\subset  \mathcal{I}^{full}\,\,\mbox{ and }\,\,\mathbf{Y}_a^+\neq \mathbf{Y}^+
	\end{equation}
\end{definition}
\subsection{Attacks with partial model knowledge}

The next result proposition shows that such attacks cannot compromise the system integrity while remaining stealthy. 
%
\begin{proposition}\label{no-stealthy-attack-identification}
{\it	
Given the NC-CPS model (\ref{eq:sys})-(\ref{eq:constraints}) subject to cyber attacks modeled as (\ref{acutator_attack_model}) and (\ref{sensor_attack_model}) and  regulated by   the state feedback law $u^c(t)=f(\tilde{y}(t))$   obtained via the \textit{ST-RHC} scheme, then a detector module $\mathbf{D},$ capable to reveal  \textit{attacks with partial model knowledge}, $\mathcal{I}^a\subset  \mathcal{I}^{full},$ is achieved as the result of  the following set-membership requirement:
\begin{equation}\label{eq:condition_attack}
\tilde{y}(t+1) \in \mathbf{Y}^+
\end{equation}
}
\end{proposition}
\begin{proof}
Under the attack free scenario hypothesis,    the current  control action $u^c(t)$ guarantees that the   one-step ahead state evolution $y^+:=A\tilde{y}(t)+Bu^c(t)+B_dd_x(t)+d_y(t)$ belongs to $\mathbf{Y}^+:$
\begin{equation}\label{set-prediction-membership}
y^+\in \mathbf{Y}^+(\tilde{y}(t),u^c(t)) ,\,\,\,\,\,\,\forall d_x(t)\in \Df_x,\,d_y(t)\in \Df_y
\end{equation}
Since  cyber attacks can occur,  two operative scenarios can arise at the next time instant $t+1:$
$$
(i)\,\,\, \tilde{y}(t+1) \notin \mathbf{Y}^+, \,\,\,\,\,\,\, (ii) \,\,\, \tilde{y}(t+1) \in \mathbf{Y}^+
$$
%
%
If (i) holds true  then the attack is 
 instantaneously detected.
Otherwise when (ii) takes place, the following arguments are exploited.  First, an attacker could modify the control signal by adding a malicious data $u^a(t)$ and, simultaneously, the detection can be avoided by infecting the effective measurement  $y(t)$ as follows:
%
$$
\mbox{Find } y^a(t): y(t)+y^a(t)\in \mathbf{Y}^+. 
$$
Because the set $\mathbf{Y}^+$ is unknown (see \textit{Definition} \ref{def:no-stealthy-attack})   such a  reasoning is not feasible. 
A second possible scenario could consist in  injecting small sized  perturbations $u^a(t)$ and $x^a(t)$ such that  
%
$$
Bu^a(t)+B_dd_x(t)\in \Df_x \mbox{ and } x^a(t)+d_y(t)\in \Df_y,
$$
Clearly,  in this case by construction    the computed  command $u^c(t+1)$    remains feasible.
\end{proof}

\subsection{Attacks with full model knowledge}\label{section:attack_full_model_knoledge}

A simple stealthy attack can be achieved by means of   the following steps:

{\it{ 
		\vspace{-0.25 cm}
		\noindent \hrulefill
		\vspace{-0.3 cm}
		\begin{center}
			\textbf{Stealthy Attack algorithm}
		\end{center}
		\vspace{-0.5 cm}
		\noindent \hrulefill

		\noindent \textbf{Knowledge:} $\mathcal{I}^{full}$
	\fontsize{10}{10}\selectfont

		\begin{algorithmic}[1]
			\State Acquire ${y}(t);$ \label{primo_step_stealthy}
			\State Estimate the  control action $\hat{u}^c(t)$ by emulating the  optimization (\ref{fun_opt_b_3})-(\ref{cond_opt_b_5}); 
			\State Compute the expected disturbance-free one-step ahead state measurement
			$
			\bar{y}^+=Ay(t)+B\hat{u}^c(t) \in \mathbf{Y}_a^+
			$
			\State Corrupt  $u^c(t)$ with an arbitrary malicious admissible  signal $u^a(t)$ such that
			$
			\hat{u}^c(t)+u^a(t)\in \mathcal{U}
			$
			\State  Corrupt the output vector $y(t+1)$, according to expected one-step state evolution $\bar{y}^+$ i.e.
			$
			y^a(t): \,\,y(t+1)+y^a(t)=: \bar{y}^+
			$
			\State $t\leftarrow t+1$, goto Step \ref{primo_step_stealthy}
		\end{algorithmic}
		\vspace{-0.35 cm}
		\hrulefill
	}} 

\noindent Note  that the above attack can never be identified by the proposed detector $\mathbf{D}$ because condition (\ref{eq:condition_attack}) is always satisfied.
As a consequence,  the only way to detect it is to increase the information available at  the defender side  $\mathcal{I}^{full}$ so that the {\it{partial model knowledge}}attack structure is re-considered:
%
\begin{equation}
\mathbf{Y}_a^+\neq \mathbf{Y}^+
\end{equation}
The key idea  traces the philosophy behind the  watermarking approach \cite{MoWeSi15}, where the defender superimposes a noise control signal  (new information not available at the defender side) in order to authenticate the physical dynamics.  In particular, a watermarking-like behavior can be straightforwardly obtained 
by  using  the  \textit{ST-RHC}  property discussed  in \textit{Remark} \ref{remark:varaible_cost_function}. 
\begin{proposition}\label{set-theoretic-watermarking}
{\it
Let  (\ref{eq:sys})-(\ref{eq:constraints}) and (\ref{acutator_attack_model})-(\ref{sensor_attack_model})  be the plant  and the FDI attack models, respectively. Let
\begin{equation}\label{cost_functions}
\mathbf{J}=\left\{J_k(\tilde{y}(t),u)\right\}_{k=1}^{N_j}
,\,\,
\mathbf{F^0}=\left\{f^0_k(\tilde{y}(t))\right\}_{k=1}^{N_j}
,\,\,N_j>1
\end{equation}
be  finite sets of cost functions and  stabilizing state-feedback control laws compatible with $\mathcal{T}^0,$ respectively.  Let $j(t):\Z_+\rightarrow [1,\ldots, N_j]$ be a  random function. 
If at each time instant $t$  the command input  $u^c(t)$ is obtained as the solution of  (\ref{fun_opt_b_3})-(\ref{cond_opt_b_5}) with   $J_{j(t)}(\tilde{y}(t),u)$  and  $f^0_{j(t)}(\tilde{y}(t))$  randomly chosen, then the anomaly detector module (\ref{eq:condition_attack}) is capable to detect  complete model  knowledge $\mathcal{I}^{full}$ attacks.
%
} 
\end{proposition}
\begin{proof}
Because the additional information $j(t)$ is not available to the attacker,  then  the following  time-varying information flow results:
\begin{equation}
\mathcal{I}^{full}(t):=
	\left\{
	 (\ref{eq:sys})\!\!-\!\!(\ref{eq:constraints}),f^0\!\!,
	\{\mathcal{T}^i\}_{i=0}^N,\,
	\tilde{y},
	\mbox{opt: }(\ref{fun_opt_b_3})\!\!-\!\!(\ref{cond_opt_b_5}),j(t) 
	\right\} \supset {I}^{full}
\end{equation}
This implies that 
$\mathcal{I}_a\subset \mathcal{I}^{full}(t)$
and,  as a consequence, a perfect stealthy attack is no longer  admissible
$$
\mathbf{Y}^+(\tilde{y}(t),u^c(j(t)))\neq\mathbf{Y}^+_a(\tilde{y}(t),u^c(t))
$$
Therefore, the detection rule (\ref{eq:condition_attack}) is effective. 
\end{proof}

\noindent Finally, by collecting the results of \textit{Propositions} \ref{no-stealthy-attack-identification}-\ref{set-theoretic-watermarking},  a solution to the   \textbf{P1}  problem is given by the following detector module:
\begin{equation}\label{detector_module}
\mbox{\textbf{Detector(D)}$(\tilde{y}(t))$:=}
\left\{
\begin{array}{l}
\!\!\!\!\mbox{attack}\quad \mbox{if}\quad \tilde{y}(t+1) \notin \mathbf{Y}^+ \\
\!\!\!\!\mbox{no attack}\quad \mbox{if}\quad \tilde{y}(t+1) \in \mathbf{Y}^+ 
\end{array}
\right.
\end{equation}
\section{Cyber-Physical countermeasures for resilient and secure control}

%
Once a  attack has been physically detected, the following cyber countermeasures can be adopted to recover an attack free scenario:
\begin{itemize}
	\item  Interrupt all the \textit{sensor-to-controller} and \textit{controller-to-actuators} communications links;
	\item  Reestablishing new secure encrypted  channels.
\end{itemize}
From a physical point of view, the prescribed actions imply that for an assigned time interval, namely $T_{encry},$ update measurements and control actions are not available at the controller and actuator sides, respectively.  Therefore, the main  challenge is:
\begin{center}
\vspace{-0.2 cm}
\textit{How we can ensure that, at least, the minimum safety requirements $x(t)\in \mathcal{X},\, u(t)\in\mathcal{U}$ are met while the communication are interrupted for $T_{ecnry}$ time instants?}
\vspace{-0.2 cm}
\end{center}
The next section will be devoted to answer to this key question.

\subsection{$\tau$-steps feasible sets and associated set-theoretic controller ($\tau$-ST-RHC)}
Let $\mathcal{T}$ be a RPI region for the plant model (\ref{eq:sys})-(\ref{eq:constraints}) subject to the induced time-delay $\tau,$ see \cite{FraTeFa15}.
Then, a family of $\tau$-steps controllable sets, $\{\mathcal{T}^i(\tau)\}_{i=1}^{N},$  can be defined as follows
\begin{equation}\label{k-steps-feasible-sets}
\begin{array}{l}
\mathcal{T}^0(\tau)  := \mathcal{T}\vspace{-0.7cm}\\
\mathcal{T}^i(\tau) :=\!\!\! \{ x \!\! \in \!\!\Xf: \exists \, u\!\in\! \Uf:\!\!\!\! \overbrace{A^k}^{A(k)}\!\!\! x\! +\! (\displaystyle \!\overbrace{\sum_{j=0}^{k-1}{\!\!\!A^j}B)}^{B(k)} u\!   \subseteq\! 
\tilde{\mathcal{T}}^{i-1}(\tau)\vspace{-0.2cm}\\
\quad \quad \quad \quad \forall\,\, k\!=\!1,\!\ldots\!,\! \tau \}
\end{array}
\vspace{-0.5cm}
\end{equation}
with 
\begin{equation}\label{p-difference-k-steps}
\left\{
\begin{array}{l}
\tilde{\mathcal{T}}_1^i(\tau)={\mathcal{T}}^i\sim B_d\Df_x\sim  A \Df_y.\\
\tilde{\mathcal{T}}_k^i(\tau)=\!\tilde{\mathcal{T}}_{k-1}^{i}\!\!(\tau)\!\sim\!\! A^{k-1}B_d\Df_x\!\sim\!  A^k \Df_y,\,\,k=2,\ldots, \tau.
\end{array}
\right.
\end{equation}
An equivalent description $\{\Xi^i(\tau)\}$ of (\ref{k-steps-feasible-sets}) can be given in terms of the extended space $(x,\,u):$ 
\begin{equation}\label{k-steps-feasible-aug-set}
\begin{array}{c}
\Xi^i(\tau)=\{ (x,\,u) \in \Xf\times\Uf : A(k)x\! +\! {B(k)} u\!   \in\! \tilde{\mathcal{T}}^{i-1}(\tau),\\
 \forall k=1,\ldots, \tau\}\\
\,\,\,\,\,\,\,\,\,\,\,\,\,\,\,\,\,\,=\displaystyle\bigcap_{k=1}^{\tau}\! \{\!(x,\,u)\!\! \in \Xf\!\!\times\Uf : A(k)x\! +\! {B(k)} u\!   \in\!\! \tilde{\mathcal{T}}^{i-1}\!(\tau)\}
\end{array}
\end{equation}
%
Hence, the sets of all the admissible state and  input vectors are simply determined as follows:
\begin{equation}\label{k-steps-feasible-x-u-set}
\mathcal{T}^i(\tau)={Proj}_{x} \Xi^i(\tau),\quad \mathcal{U}^i(K)={Proj}_u \Xi^i(\tau)
\end{equation}
where ${Proj}_{(\cdot)}$ is the projection operator \cite{Bl08}.
\begin{proposition}\label{property:control_action_feasible_k_steps}
{\it 
Let  the set sequences  $\{\Xi^i(T_{encry})\}_{i=1}^{N},$ $\{\mathcal{U}^i(T_{encry})\}_{i=1}^{N},$ $\{\mathcal{T}^i(T_{encry})\}_{i=1}^{N}$ be given. Under the attack free scenario hypothesis  ($u^a(t)\equiv0,\, y^a(t)\equiv 0$), 
 the control action $u^c(t),$ computed by means of the following convex optimization problem
	\begin{equation}\label{new_opt_k_steps_fun}
	{u}^c(t) =  \arg \min_u J_{j(t)}(y(t),u) \quad s.t.
	\end{equation}
	\begin{equation}\label{new_opt_k_steps_constr}
	[y(t),\,u]\in \Xi^i(T_{encry}),\quad  u \in \Uf^i(T_{encry}) 
	\end{equation}
	%
and consecutively applied to (\ref{eq:sys}) for $T_{ecnry}$ time instants, ensures: \textit{i)} constraints fulfillment; \textit{ii)} state trajectory confinement, i.e. $x(t+k)\in \mathcal{T}^{i(t)-1}(T_{encry}),\,\forall k=1,\ldots, T_{encry},$ regardless of any $d_x(\cdot)\in \mathcal{D}_x$ and $\,d_y(\cdot)\in \mathcal{D}_y$ realizations and  any  cost function $J_{j(t)}(y(t),u).$
}
\end{proposition}
%
\begin{proof}
By construction of   (\ref{k-steps-feasible-sets})-(\ref{k-steps-feasible-x-u-set}),   it is always guaranteed that, if $x(t)\in \mathcal{T}^{i(t)}(T_{encry})$, the optimization  (\ref{new_opt_k_steps_fun})-(\ref{new_opt_k_steps_constr}) is feasible  and an admissible  $u^c(t)$ there exists. Moreover  if for $T_{encry}$ time instants the command ${u}^c(t)$  is consecutively applied to (\ref{eq:sys}),  one has that
$$
x(t+k)\!=\!A^k x(t)\!+\!\!\!\sum_{j=0}^{k-1}\!\!\!(A^jB)u^c(t)\!+\!\!\!\underbrace{\sum_{j=0}^{k-1}\!\!\!(A^jB_d)d_x(j)}_{unknown},\,k=1,\!\ldots,\!T_{encry}
$$
Then in virtue of  (\ref{k-steps-feasible-sets}),  the disturbance-free evolution $\bar{x}(t+k)$ is
$$
\bar{x}(t+k)\in \tilde{\mathcal{T}}^{i-1}_k(T_{encry}),\,\forall\,k=1,\!\ldots,\!T_{encry}
$$
and
the following implications hold true
$$
\begin{array}{c}
\forall d_x(t+k)\in \Df_x, d_y(t+k)\in \Df_y\!\!\!\! \implies\!\!\!\! x(t+k+1)\!\!\!\in \!\!{\mathcal{T}}^{i-1}_k(T_{encry})\\
k=0,\ldots, T_{encry}-1.
\end{array}
$$
%
\end{proof}	
\begin{remark}\label{remark:k-step-sets-computation}
The optimization  (\ref{new_opt_k_steps_fun})-(\ref{new_opt_k_steps_constr}) is solvable in polynomial  time and the required computational burdens   are irrespective 
of the number of steps $T_{encry}.$  Further  details on the computation of  the $\tau$-steps controllable sets  can be found in 
\cite{Bl08},\cite{KuVa97} for comprehensive tutorials and  \cite{MPT3},\cite{Ell_toolbox} for available  toolboxes.  
%
\hfill $\Box$
\end{remark}
In the sequel, the control strategy arising from the solution of (\ref{new_opt_k_steps_fun})-(\ref{new_opt_k_steps_constr}) will be named   \textbf{$\tau$-ST-RHC} controller.  Note that it is not able to address  all the  attack scenarios, because  if the more recent   action $u^c(t)$ 
has been corrupted,  the Proposition \ref{property:control_action_feasible_k_steps} statement becomes no longer valid. In such a case,  the  defender can only  use a \textit{smart} actuator module that locally, by means of simple security checks, is able to understand if the most recent command input  is malicious.

\subsection{Pre-Check and Post-Check firewalls modules}\label{pre-post-check}
In what follows,  two complementary modules,  hereafter named \textbf{Pre-Check} and \textbf{Post-Check}, are introduced, see Fig. \ref{fig:distribuited_control_architecture}.
The reasoning behind them  is to passively detect attacks before they could harm the plant. In particular,  such modules are  in charge of checking the following state and input  set-membership requirements:
%
\begin{equation}\label{pre_check}
\mbox{\textbf{Pre-Check}$(i(t))$:=}
\left\{
\begin{array}{l}
\!\!\!\!true\quad \mbox{if}\,\,\,\tilde{u}(t)\in \{\mathcal{U}^{i}(T_{encry})\}_{i=1}^{i(t)}\\
\!\!\!\!false\quad \mbox{otherwise}
\end{array}
\right.
\end{equation}
%
\begin{equation}\label{post_check}
\mbox{\textbf{Post-Check}$(i(t))$:=}\!
\left\{
\begin{array}{l}
\!\!\!\!\!\!true\,\,\, \mbox{if}\,\,\,y(t)\!\in \!\!\{\mathcal{T}^{i}(T_{encry})\}_{i=1}^{i(t)}\!\!\!\!\\
\!\!\!\!\!\!false\,\,\, \mbox{otherwise}
\end{array}
\right.
\end{equation}
Conditions (\ref{pre_check})-(\ref{post_check})  check if the received $\tilde{u}(t)$ and  the measurement $y(t)$ 
 are ``coherent"
with the expected set level $i(t).$  If one of  these tests  fails, then a warning flag  is sent to the actuator and an attack is locally claimed.\\  
%
In response to the received flag, different actions are  performed by the actuator:
if the \textbf{Pre-Check} fails, $\tilde{u}(t)\notin \{\mathcal{U}^{i}(T_{encry})\}_{i=1}^{i(t)},$ then  $\tilde{u}(t)$ is discarded and the admissible stored input, hereafter named $u_{-1}:=\tilde{u}(t-1),$  applied;
if the \textbf{Post-Check} fails, $y(t)\notin \{\mathcal{T}^{i}(T_{encry})\}_{i=1}^{i(t)},$ then an harmful command $\tilde{u}(t-1)$ has been applied bypassing the \textbf{Pre-Check} control.  As a consequence,  $u_{-1}$ cannot be used at the next time instants. In this circumstance,   a possible solution consists in applying the zero input  $u(t)\equiv 0$   until safe communications are reestablished. The latter gives rise to the following problem:

\textit{How can one ensure that the open-loop system subject to $u(t)\equiv 0$
 fulfills the prescribed constraints (\ref{eq:constraints}) and is UUB?
}

\noindent The following  developments provide a formal solution.


\noindent Let  denote with  $\mathcal{T}^{i_{max}},\,\,\,i_{max}\leq N$  the maximum admissible  set computed as follows
%
\begin{equation}\label{safe_free_evolution}
\begin{array}{c}
i_{max}=\displaystyle \max_{i\leq N} \,\,i \quad s.t.\\
\underbrace{A^k\mathcal{T}^i(T_{encry}) \oplus \displaystyle \sum_{j=0}^{k-1}A^{j}B_d\Df_x}_{\mbox{first term}}\oplus\!\!\!\!\!\! \underbrace{A^{k-1}B\mathcal{U}}_{\mbox{second term}}\!\!\!\!\!\!\!\! \in\!\!\!\!\!\!\!  \displaystyle \bigcup_{j=1}^{\min(N, i+T_{viol})}\!\!\!\!\!\!\!\!\!\!\!\mathcal{T}^j(T_{encty})\\
i=1,\ldots, i_{max},\,\,k=1,\ldots, {T_{encry}}.
\end{array}
\end{equation}
Note that  the first term represents the autonomous  state evolution of (\ref{eq:sys}), whereas the second one takes care of an  unknown input
 $u\in \mathcal{U}.$ 
Moreover, the upper bound $\min(N, i_{max}+T_{viol})$ is complying with the Assumption \ref{ass2}, where it is supposed that, after the recovery phase, a new attack could only occurs after $T_{viol}$ time instants.


The reasoning behind the introduction of $\mathcal{T}^{i_{max}}$ concerns with the following feasibility retention arguments. When data (state measurements and control actions) flows are interrupt the {\bf{NC-CPS}} model (\ref{eq:sys}) evolves in an open-loop fashion under a zero-input action. Therefore the computation (\ref{safe_free_evolution}) guarantees that, starting from any initial condition belonging to  $\displaystyle\bigcup_{i=1}^{i_{max}}(\, \mathcal{T}^i(T_{encry})$  the resulting $T_{encty}$ -step ahead state  predictions of (\ref{eq:sys}) are embedded in the worst case within the DoA $\displaystyle\bigcup_{i=1}^N\, \mathcal{T}^i(T_{encry}).$ 
\begin{proposition}\label{free_evolution_bounded_stability}
{\it
Let  $\{\mathcal{T}^i(T_{encry})\}_{i=1}^{N}$  and $\mathcal{T}^{i_{max}}(T_{encry})$  be the $\tau-$step ahead controllable set sequence and the maximum admissible set, respectively.
If the plant model (\ref{eq:sys})   is operating under  a free  attack  scenario  and the state evolution $x(\cdot)$ is confined to $\displaystyle\bigcup_{i=1}^{i_{max}}\, \mathcal{T}^i(T_{encry}),$
then the zero-input  state evolution of  (\ref{eq:sys}) will be confined to  $\displaystyle\bigcup_{i=1}^N\, \mathcal{T}^i(T_{encry})$ irrespective of  any cyber attack occurrence and disturbance/noise realizations. 
}
\end{proposition}
\begin{proof}
Constraints fulfillment and UUB  trivially follow because $0_m\in \mathcal{U}$ and  $\bigcup_{i=1}^{\min(N, i_{max}+T_{viol})}\{\mathcal{T}^i(T_{encry})\}\subseteq \mathcal{X}$.
\end{proof}

\subsection{The RHC algorithm}
The above developments allow to write down the following computable scheme.

{\it{ 
		\noindent \hrulefill
		\vspace{-0.3 cm}
		\begin{center}
			\textbf{Actuators} Algorithm
		\end{center}
		\vspace{-0.5 cm}
		\noindent \hrulefill

		\noindent \textbf{Input:} $\tilde{u}(t),$ \textbf{Pre-Check}, \textbf{Post-Check}\\
		\noindent \textbf{Output:} The applied control input $u$, the expected set-level $\hat{i}$ \\
		\textbf{Initialization:} $\hat{i}=i(0),$ $u_{-1}=u^c(0)$
			\fontsize{10}{9}\selectfont
		\begin{algorithmic}[1]
			\If{\textbf{Pre-Check}($\hat{i}$)==true \& \textbf{Post-Check}($\hat{i}$)==true}
					\If{$\tilde{u}(t)=\emptyset$} \Comment{Command not received}
					\State $u(t)=u^{-1}$; \Comment{Apply previous command} \label{previous_command}
					\Else{} 
					\State $u(t)=\tilde{u}(t),$  $\hat{i}=\hat{i}-1;$ \Comment{estimated  set-level update}
					\EndIf
			\Else{} \If{\textbf{Pre-Check}($\hat{i}$)==false} $u(t)=u_{-1}$
			\Else{} $u(t)=0;$  \Comment{Attack locally detected and free-evolution} \label{free-evol}
					\EndIf
			\EndIf
			\State $u^{-1}\leftarrow u(t)$
			\State $t\leftarrow t+1,$ goto Step 1
		\end{algorithmic}
		\hrulefill
	}} 

{\it{ 
		\noindent \hrulefill
		\vspace{-0.3 cm}
		\begin{center}
			\textbf{$\tau$-ST-RHC} Controller Algorithm (Off-line)
		\end{center}
		\vspace{-0.5 cm}
		\noindent \hrulefill

		\noindent \textbf{Input:} $T_{encry}$ \\
		\noindent \textbf{Output:} $\left\{\Xi^{i}\right\}_{i=0}^{N}(T_{encry}),\left\{\mathcal{T}^{i}\right\}_{i=0}^{N}(T_{encry}),$ $\left\{\mathcal{U}^{i}\right\}_{i=0}^{N}(T_{encry}),$ $i_{max}$
		\fontsize{10}{9}\selectfont
		\begin{algorithmic}[1]
			\State Compute a RPI region $\mathcal{T}^0$
			\State Compute the families  of $T_{encry}-$steps controllable sets $\left\{\Xi^{i}\right\}_{i=0}^{N}(T_{encry}),\left\{\mathcal{U}^{i}\right\}_{i=0}^{N}(T_{encry}),\left\{\mathcal{T}^{i}\right\}_{i=0}^{N}(T_{encry})$ by resorting to recursion (\ref{k-steps-feasible-aug-set}) and to the projection (\ref{k-steps-feasible-x-u-set})
			\State Determine the maximum index $i_{max}$ satisfying (\ref{safe_free_evolution})
			\State Collect $N_j>1$ cost functions (\ref{cost_functions}) and  terminal control law $f^0_{j(t)}(\tilde{x}(t))$
		\end{algorithmic}
		\hrulefill
	}} 

{\it{ 
		\noindent \hrulefill
		\vspace{-0.3 cm}
		\begin{center}
			\textbf{$\tau$-ST-RHC} Controller Algorithm (On-line)
		\end{center}
		\vspace{-0.5 cm}
		\noindent \hrulefill

		\noindent \textbf{Input:} $\tilde{y}(t),$ $\left\{\Xi^{i}\right\}_{i=0}^{N}(T_{encry}),$ $\left\{\mathcal{T}^{i}\right\}_{i=0}^{N}(T_{encry}),$ $\left\{\mathcal{U}^{i}\right\}_{i=0}^{N}(T_{encry}),$ \textbf{Detector}$(\tilde{y}(t)),$ $i_{max},$ $\mathbf{J}$\\
		\noindent \textbf{Output:} Computed command $u^c(t)$\\
		\textbf{Initialization:} status=no attack, timer=0, encrypted communication channels, initialize \textbf{Detector}, \textbf{Pre-Check},  \textbf{Post-Check}, \textbf{Actuator} modules\\
		\textbf{Feasibility start condition:} $\exists i<(i_{max}+T_{viol})\leq N:\, x(0)\in \bigcup_{i=0}^{max(N,i_{max}+T_{viol})}\{\mathcal{T}^{i}(T_{encry})\}$ 
		\fontsize{10}{9}\selectfont
		\begin{algorithmic}[1]
			
			\If{status==no attack} \Comment{Start Automa}
				\If{\textbf{Detector}$(\tilde{x}(t))$==attack}  status=attack
				\EndIf
			\Else{} \label{start_encry}
				\If{status==attack} \Comment{Wait channel encryption} 
					\If{timer$<T_{encry}$} timer=timer+1;
					\Else{} 
					\State 
					\begin{varwidth}[t]{0.85\columnwidth}
					Re-initialize all modules; status==no attack; timer=0;
					\end{varwidth}
					\EndIf
				\EndIf \label{end_encry}
			\EndIf \Comment{End Automa}
			
			\If{status==no attack} 
			\State Find 
			\vspace{-0.25 cm}
			$$
			i(t)=\arg\min_i:\, \tilde{y}(t)\in \mathcal{T}^{i}(T_{encry})
			$$ 
			\vspace{-0.25 cm}
			\State 	Randomly choose the selection index $\bar{j}=j(t)$;
			\If{$i(t)==0$}  $u^c(t)=f^0_{\bar{j}}(\tilde{y}(t))$
			\Else{}
			\State 
			\begin{varwidth}[t]{0.80\columnwidth}
				Compute  $u^c(t)$ by solving  (\ref{new_opt_k_steps_fun})-(\ref{new_opt_k_steps_constr}) with cost function $J_{\bar{j}}(\tilde{y}(t),u)$;
			\end{varwidth}
			\EndIf
			
			\State Send $u^c(t)$ to the actuators;
			\Else{}
				\State Interrupt all the communications
				\State Reestablish  encrypted communication channels 
			\EndIf
		\State $t\leftarrow t+1$, goto Step 1
		\end{algorithmic}
		\hrulefill
	}} 

\begin{remark}\label{remark:stima_set_level_attuatore}
It is important to underline that  \textbf{Pre-Check} and \textbf{Post-Check} modules need  the current set-level $i(t). $ Unfortunately,  this information   cannot be transmitted  because it could be modified  by some attackers. To overcome such a difficulty, the  estimate $\hat{i}(t)$ provided by the \textbf{Actuator} unit is used. Note that  $i(t)$ and  $\hat{i}(t)$  are synchronized  at  the initial ($t=0$)   and at each recovery phase time instants, while in all the other situations it is ensured that such signals are compatible, i.e. $\hat{i}(t)\geq  i(t), \, \forall t \geq 0.$
\hfill $\Box$
\end{remark}
\begin{theorem}\label{teorema:proof_P1_P2_problem}
\it{
Let  $\{\Xi^i(T_{encry})\}_{i=1}^{N},$ $\{\mathcal{U}^i(T_{encry})\}_{i=1}^{N},$ $\{\mathcal{T}^i(T_{encry})\}_{i=1}^{N},$ be non empty controllable set sequences and 
$$
x(0)\in \bigcup_{i=0}^{max(N,i_{max}+T_{viol})}\{\mathcal{T}^{i}(T_{encry})\}
$$
Then, the proposed set-theoretic control architecture (\textbf{$\tau$-ST-RHC}  Controller, \textbf{Detector}, \textbf{Pre-Check} and \textbf{Post-Check}) always guarantees constraints  satisfaction   and Uniformly Ultimate Boundedness    for all admissible attack scenarios  and disturbance/noise realizations.  
}
\end{theorem}
\begin{proof}
%
The proof straightforwardly  follows by ensuring that under any admissible attack scenario the following requirements hold true:
\begin{itemize}
\item the on-line optimization problem (\ref{new_opt_k_steps_fun})-(\ref{new_opt_k_steps_constr}) is  feasible and  the state trajectory $x(t)$ is confined to
$\bigcup_{i=0}^{N}\{\mathcal{T}^{i}(T_{encry})\};$

\item any attack free scenario can be recovered in at most  $T_{encry}$ time instants. 
\end{itemize}
As shown in Section \ref{pre-post-check}, the worst case scenario arises when the attacker can successful inject a malicious input that simulates  a stealthy attack. 
First, in virtue of the actions of  the \textbf{Pre-Check} and \textbf{Actuator} modules, the input  constraints $u(t)\in \mathcal{U}$ are always fulfilled. 
Then,  the \textbf{Post-Check} module ensures that, whenever the state trajectory diverges within the $T_{encry}$-steps ahead controllable set sequence, a recovery procedure starts and an admissible zero-input state evolution takes place, see Proposition \ref{free_evolution_bounded_stability}.

\end{proof}

\section{NUMERICAL EXAMPLE}

We consider the continuous-time model  \cite{BlMI00}
$$
\left[
\begin{array}{c}
\dot{x}_1(t)\\\dot{x}_2(t) 
\end{array}
\right]\! =\!
\left[
\begin{array}{c c}
1 & 4\\ 0.8 & 0.5
\end{array}
\right] 
\left[
\begin{array}{c}
x_1(t)\\x_2(t) 
\end{array}
\right] 
\!\!+\!\!
\left[
\begin{array}{c}
0\\1 
\end{array}
\right] 
u(t)
\!\!+\!\!
\left[
\begin{array}{c}
1\\1 
\end{array}
\right] 
d_x(t)
$$
subject to
$$
	|u(t)| \leq 5, |x_1(t)| \leq 2.5,  |x_2(t)| \leq 10, |d_x (t)|\leq 0.05
$$
The continuous time system  has been discretized by means of Forward Euler-method with sampling time $T_s = 0.02\,sec.$ According to \textit{Assumption} \ref{ass1}-\ref{ass3}, we consider a reliable encrypted communication medium where $T_{encry}=4$ time steps ($ 0.08sec$) and  $T_{viol}=5$  time steps ($0.1sec$).

\noindent First, 
the following polyhedral families of $T_{encry}-$steps controllable sets are computed (see Fig. \ref{fig:sets_and_trajectory}):
$$
\left\{\Xi^{i}\right\}_{i=0}^{60}(T_{encry}),\,\,\,\left\{\mathcal{U}^{i}\right\}_{i=0}^{60}(T_{encry}),\,\,\,\,\left\{\mathcal{T}^{i}\right\}_{i=0}^{60}(T_{encry})
$$ 
and the maximum safe index set $i_{max}=45$ has been determined. 

\begin{figure}
\centering
\includegraphics[width=0.6\linewidth]{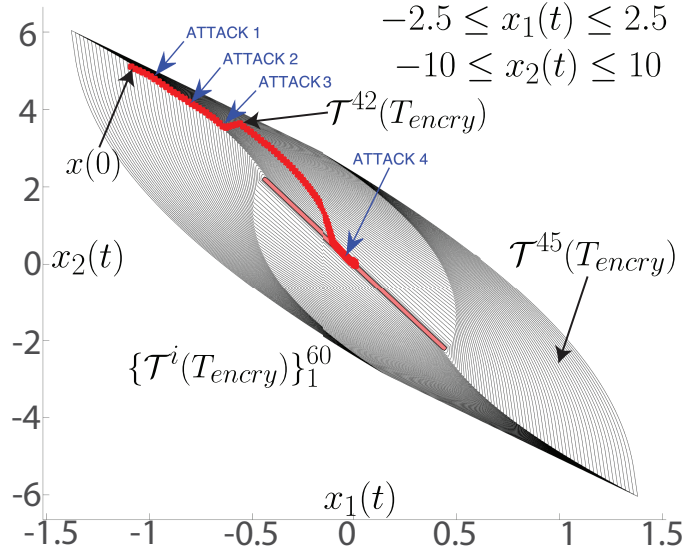}
\caption{$\left\{\mathcal{T}^{i}\right\}_{i=0}^{60}(T_{encry})$ family (black polyhedra) and state trajectory (red solid line). Blue arrows point to the  current system state vector at the beginning of each attack scenario. }
\label{fig:sets_and_trajectory}
\end{figure}

The following simulation scenario is considered:\\
{\it{Starting from the initial condition   $x(0)=[-1.09,\,5.11]^T\in \mathcal{T}^{45}(T_{encry}),$ regulate the state trajectory to zero 
regardless of any admissible attack  and disturbance realization and satisfy the prescribed constraints.}} 

\noindent In the sequel, the following sequence of attacks is considered:
\begin{itemize}
	\item Partial model knowledge attacks - (Attack 1) DoS attack on the \textit{controller-to-actuator} channel; (Attack 2) DoS attack on the \textit{sensor-to-controller} channel; (Attack 3) FDI attack on the \textit{controller-to-actuator} channel.

	\item FDI Full model knowledge attacks - (Attack 4). By following the \textbf{Stealthy Attack Algorithm}  of Section \ref{section:attack_full_model_knoledge}, 
	the attacker, 
	 tries to impose   the  malicious control action 
	$$
	\begin{array}{c}
	\tilde{u}(t)=\displaystyle \arg \max_u |Ax+Bu|,\,\,\,\,s.t. \\
	Ax+Bu\in \tilde{\mathcal{T}}^0,\,\,\,u\in \mathcal{U}^0\\
	\end{array}
	$$
	with the aim  to keep the state trajectory  as far as possible from the equilibrium and to avoid the \textbf{Post-Check} detection  by embedding the dynamical plant behavior within the terminal region.
\end{itemize}
%
\begin{figure}
\centering
\includegraphics[width=1\linewidth]{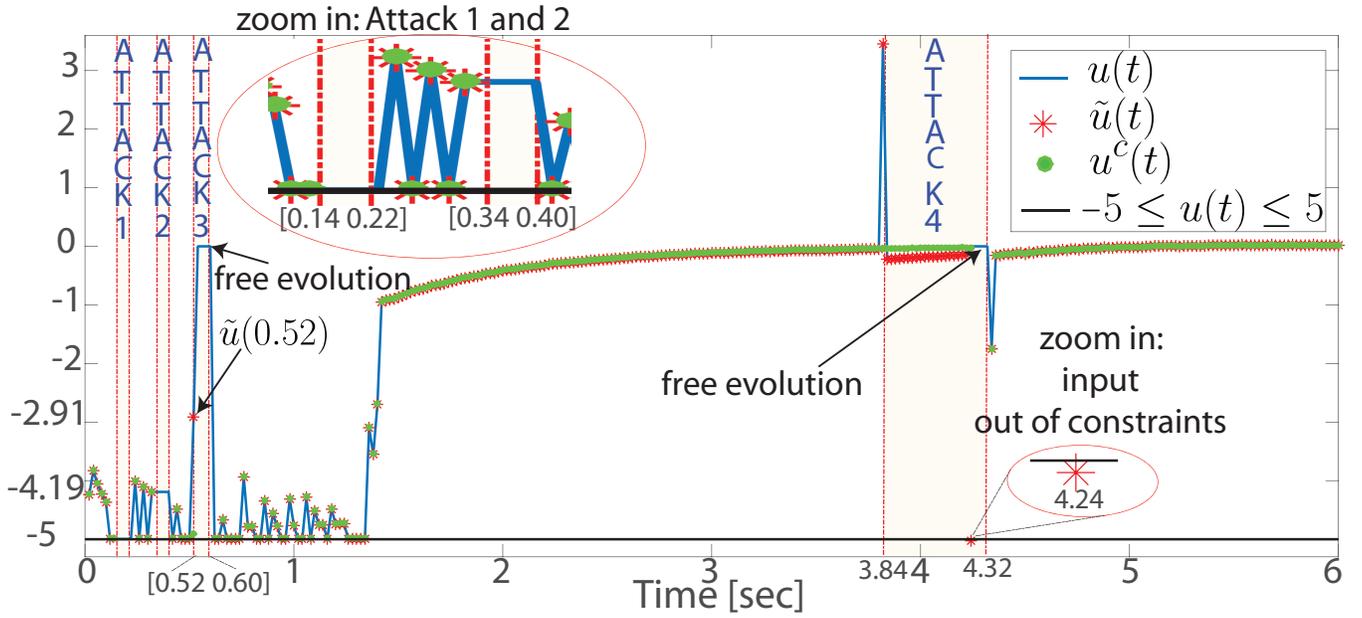}
\caption{Command inputs: actuator output $u(t),$ corrupted control action $\tilde{u}(t),$ controller command $u^c(t).$  }
\label{fig:inputs_and_attacks}
\end{figure}
\begin{figure}
\centering
\includegraphics[width=1\linewidth]{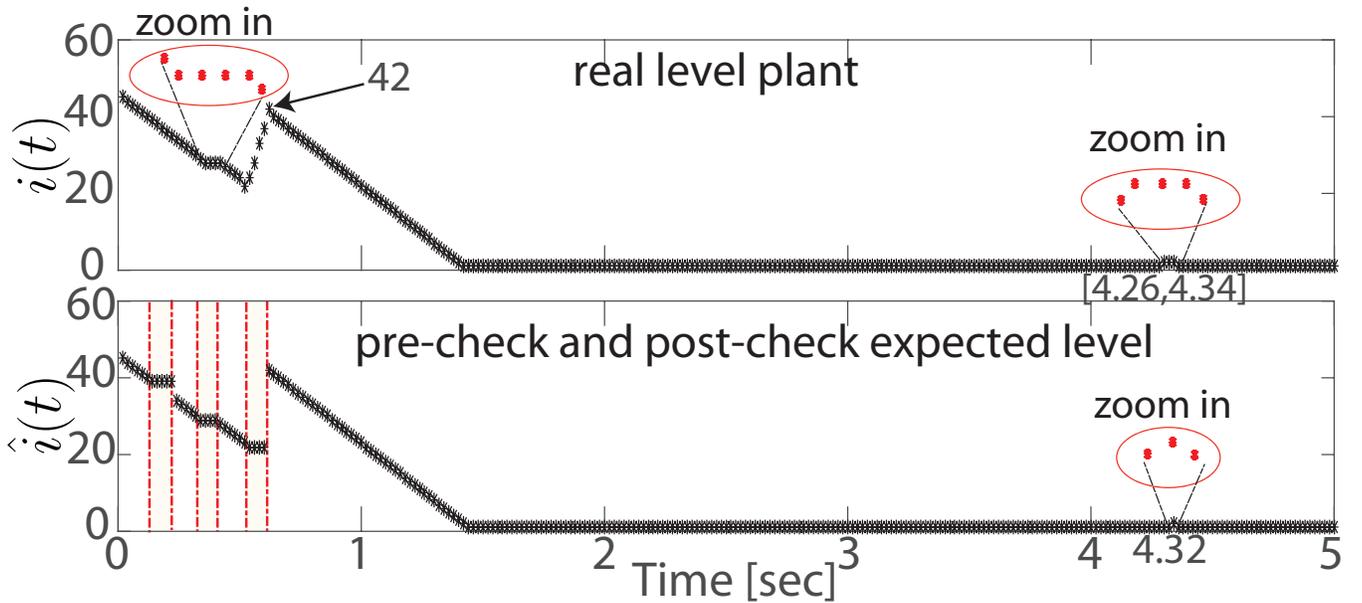}
\caption{State set-membership levels: real plant level $i(t)$ (top figure),   Pre-Check and Post-Check estimate level $\hat{i}(t)$ (bottom figure). }
\label{fig:set_levels}
\end{figure}
\begin{figure}
	\centering
	\includegraphics[width=1\linewidth]{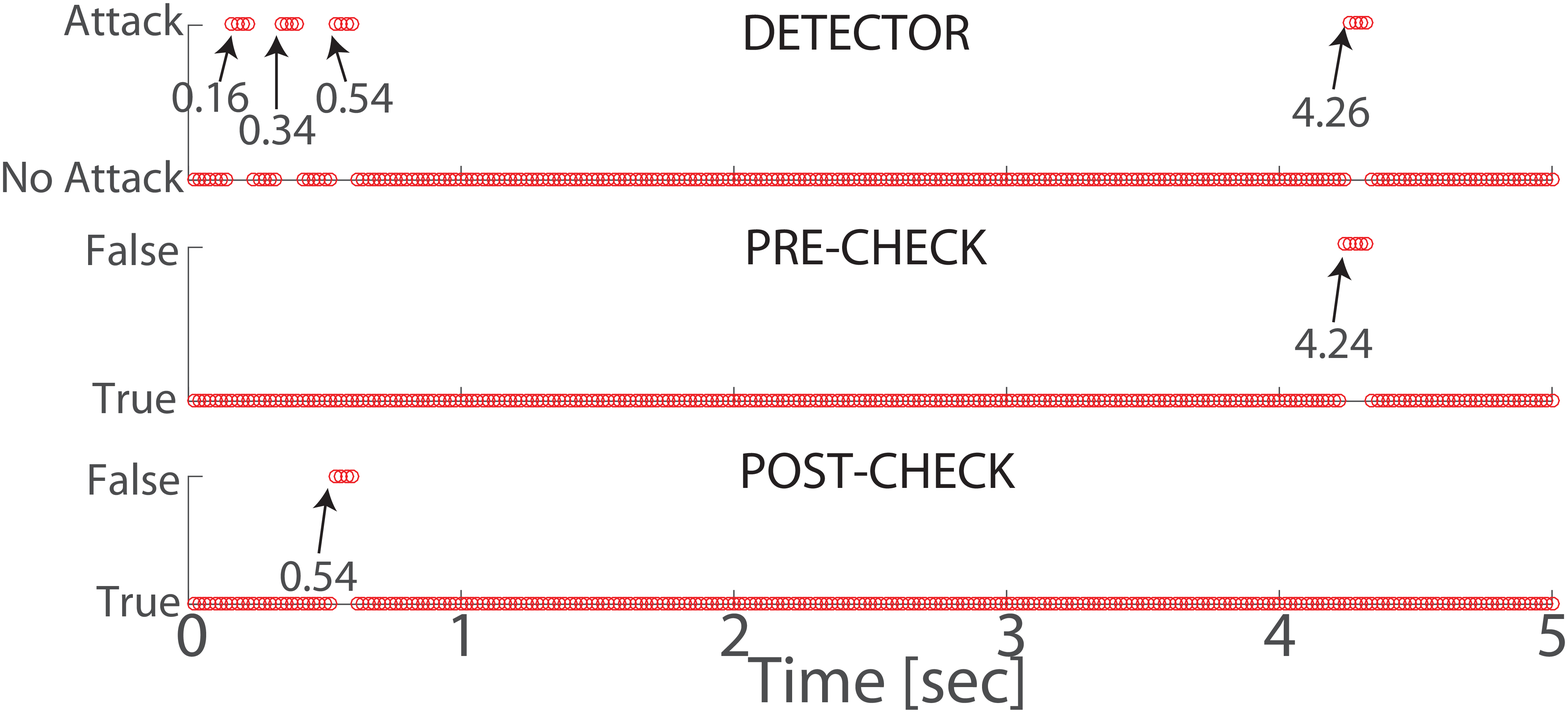}
	\caption{Detector, Pre-Check and Post-Check flag signals.}
	\label{fig:status_new}
\end{figure}
%
First, it is interesting to underline that the state trajectory is confined within $\left\{\mathcal{T}^{i}\right\}_{i=0}^{60}(T_{encry})$ and  asymptotically converges to the origin  when an attack free scenario is recovered  ($t= 4.32 sec$). 

\noindent-({\bf{Attack 1}}) Starting from  $t=0.14 sec,$  
the actuator do  not receive new packets. According to the \textbf{Actuators} algorithm (Step \ref{previous_command}) the most recent available command ($u(t)=u^c(0.12)=4.95$) can be  applied since both Pre-Check and Post-Check conditions are satisfied. At $t=0.16 sec,$ the Detector identifies the attack (see Fig. \ref{fig:status_new}) because 
%
$$
\tilde{y}(0.16)\!\!\notin\!\! \mathcal{Y}^+\!\!\!\!=\!\{\!z\in {\rm R}^n\!:\! \exists d\in \mathcal{D}, z\!=\!A\tilde{x}(0.14)\!+\!Bu^c(0.14)\!+\!B_d d\}
$$
As prescribed  in   Steps \ref{start_encry}-\ref{end_encry} of the \textbf{$\tau$-ST-RHC} algorithm,  the existing communications are interrupted and  the procedure to reestablish new encrypted channels  started. At $t=0.24 sec,$ the encryption procedure ends and all the modules  re-initialized. It is worth to notice that  neither Pre-Check or Post-Check modules trigger a False Input or False output events. This is due to the fact that the most recent  control action was not corrupted and,  by construction, it ensures that the state trajectory remains confined within the current controllable set   for the successive  $T_{encry}$ time instants, 
see Fig. \ref{fig:set_levels}.

\noindent-({\bf{Attack 2}})  Starting from  $t=0.34 sec$ 
the Controller does not receive update state measurements and   the {\bf{Detector}}  triggers an attack event. As a consequence,  the network is disconnected and the actuator does not receive new control actions and the  available command $u(t)=u^c(0.32)=-4.19$ is applied, see Fig. \ref{fig:inputs_and_attacks}. 
 At  $t=0.22 \,sec,$ the attack free scenario is recovered.

\noindent-({\bf{Attack 3}}) At $t=0.52 sec$ 
the attacker injects a signal $u^a(0.52)=2$ that corrupts 
the current input $u^c(0.52)=-4.91$ as indicated in (\ref{acutator_attack_model}). Therefore, the actuator receives  $\tilde{u}(0.52)=-2.91$ that is still admissible as 
testified by  the {\bf{Pre-Check}} unit. The main consequence of the latter is that  $x(0.54) \in \mathcal{T}^{20}$ while the expected set-membership condition should 
have to be  $\mathcal{T}^{18}:$    the Post-Check module and the Detector trigger an attack event and the Controller blocks all the communications. 
From now on, the {\bf{Actuator}} logic imposes a zero-input state evolution, 
see Step \ref{free-evol}. Although during the channel encryption phase (the encryption procedure ends at $t=0.60 sec$) the  set-membership index increases (see Fig.\ref{fig:set_levels}), this does not compromise feasibility retention because $x(0.52)\in \{\mathcal{T}^i(T_{encry})\}_{i=0}^{imax}$ and the zero-input state evolution
will remain confined within the domain of attraction, i.e. $x(0.60)=[-0.58,\, 3.61]^T\in \mathcal{T}^{42}(T_{encry}),$ see Fig. \ref{fig:set_levels}.
	
\noindent-({\bf{Attack 4}}) At $t=3.84 sec,$  with $x(3.83)\in \mathcal{T}^0(T_{encry}),$ an  FDI attack corrupts both the communication channels, see Fig. \ref{fig:sets_and_trajectory},\ref{fig:set_levels}. \\
The attacker is capable to remain stealthy until $t=4.24\,sec$ and to manipulate the plant input and outputs.  This unfavorable phenomenon is due to the fact that
it is not possible   to discriminate between the attack and the disturbance/noise realizations $d_x(t)$ and $d_y(t),$  i.e.  $\tilde{y}(t)\in \mathcal{Y}^+, \forall t\in [3.84,\, 4.22]sec.$\\
%
%
At $t=4.24\,sec,$ a different behavior occurs in response to  $\tilde{u}(4.24)=-5.025\notin \mathcal{U}^0$  with the Pre-Check module that triggers an anomalous event arising when the attacker  
tries to impose  $\tilde{u}(4.24)=-4.993$  as the current input.
Specifically, the attacker determines the  estimate $\hat{u}^c(4.24)$ and modifies the control action as follows $\hat{u}^c(4.24)+u^a(4.24)=\tilde{u}(4.24).$ Since a time-varying index is exploited in (\ref{new_opt_k_steps_fun})-(\ref{new_opt_k_steps_constr}), the  estimate  $\hat{u}^c(4.24)=-0.032$ is numerically different from the effective control action $u^c(4.24)=-0.059.$ Therefore  $\tilde{u}(4.24) \notin \mathcal{U}^0$ and, as a consequence, the attack is detected. 
For the interested reader, simulation demos are available at the following two web links: \textbf{J fixed}:  \url{https://goo.gl/8CQ4b8},  \textbf{J random}: \url{https://goo.gl/DQhOxB}


\section{Conclusions}
In this paper a control architecture devoted to  detect and mitigate cyber attacks affecting networked constrained CPS is presented. The resulting control scheme, which takes mainly  advantage from  set-theoretic concepts, provides a formal and robust cyber-physical approach against the DoS and FDI   attack classes.  Constraints satisfaction and Uniformly Ultimate Boundndeness  are formally proved  regardless of any admissible attack scenario occurrence.  Finally, the simulation section allows to show the effectiveness and applicability of the proposed strategy under severe cyber attack scenarios.

\end{document}